%% file: arXiv2.tex
\documentclass[12pt]{article}

\usepackage[margin=1in]{geometry}

\usepackage{setspace}

\usepackage{graphicx}

\usepackage{enumitem}

\usepackage{microtype}
\usepackage{subfigure}
\usepackage{booktabs}
\usepackage{comment}
\usepackage{wrapfig}
\usepackage{arydshln}
\usepackage{enumitem}
\usepackage{multirow}
\usepackage{url}
\usepackage{natbib}
\usepackage{authblk}

\RequirePackage[colorlinks,citecolor=blue,linkcolor=blue,urlcolor=blue,pagebackref]{hyperref}

\usepackage{amsmath}
\usepackage{amssymb}
\usepackage{mathtools}
\usepackage{amsfonts}
\usepackage{bm}
\usepackage{bbm}

\input{Macros/algorithm.tex}
\input{Macros/math.tex}
\input{Macros/theorems.tex}
\input{Macros/misc.tex}

\usepackage[capitalize,noabbrev]{cleveref}

\allowdisplaybreaks

\input{Head/Title}
\input{Head/AuthorInf}

\begin{document}

\maketitle

\input{Main.tex}

\end{document}

%% file: Macros/algorithm.tex
\usepackage{algorithm}
\usepackage{algorithmic}

%% file: Macros/math.tex
\def\1{\bm{1}}

\DeclareMathAlphabet{\mathsfit}{\encodingdefault}{\sfdefault}{m}{sl}
\SetMathAlphabet{\mathsfit}{bold}{\encodingdefault}{\sfdefault}{bx}{n}

\def\calF{{\mathcal{F}}}

\def\calH{{\mathcal{H}}}
\def\calI{{\mathcal{I}}}
\def\calJ{{\mathcal{J}}}

\def\calQ{{\mathcal{Q}}}

\def\calT{{\mathcal{T}}}
\def\calU{{\mathcal{U}}}

\newcommand{\E}{\mathbb{E}}



%% file: Macros/theorems.tex
\usepackage{amsthm}

\theoremstyle{plain}

\newtheorem{theorem}{Theorem}[section]

\newtheorem{corollary}[theorem]{Corollary}
\newtheorem{proposition}[theorem]{Proposition}

\newtheorem{definition}{Definition}[section]
\newtheorem{assumption}{Assumption}[section]

\newtheorem*{remark}{Remark}

%% file: Macros/misc.tex
\usepackage[textsize=tiny]{todonotes}
\usepackage{multirow}
\usepackage{wrapfig}
\usepackage{subfigure}
\usepackage{tabularx}

\usepackage{xcolor}
\newcount\Comments  
\newcommand{\kibitz}[2]{\ifnum\Comments=1\textcolor{#1}{#2}\fi}

\usepackage{listings}

\usepackage{tabularx}

\usepackage{tikz}
\usetikzlibrary{arrows.meta,positioning,calc}

%% file: Head/Title.tex
\title{GTDD: Generative Test-Driven Development for AI Coding Agents with Adversarial Testing}

%% file: Head/AuthorInf.tex
\author{Masahiro Kato\thanks{Email: \texttt{mkato-csecon@g.ecc.u-tokyo.ac.jp}}$\,$}

\affil{Mizuho-DL Financial Technology Co., Ltd.}

\date{\today}

%% file: Main.tex
\begin{abstract}
Test-driven development gives AI coding agents executable requirements for implementing software. Because these agents can adapt their implementations to the examples they observe, passing a predetermined collection of tests can leave substantial parts of the intended behavior unimplemented. We propose Generative Test-Driven Development (GTDD), a formulation of test-driven development in which a separate testing agent generates new inputs after each candidate implementation is fixed, using a human-specified behavioral contract and the feedback from earlier rounds. A trusted evaluator checks these inputs, returns reduced counterexamples to the coding agent, and saves them for regression testing, so development continually confronts failures beyond the initial examples. We characterize the evidence that this process provides through a finite-population analysis of false acceptance under adaptive candidate selection. The resulting bounds quantify how test visibility and repeated feedback affect acceptance, and show that fresh random audits after candidate commitment control false acceptance across development rounds. In a paired experiment on a stateful key-value store, both policies that regenerated tests during development ended with lower mean failure rates than the policy whose tests were generated once by the same language model, and giving the tester the candidate's source produced no detectable additional improvement. Further conditions requesting equal numbers of tests did not isolate any single feature of the policies as the source of this difference. GTDD combines this adaptive development feedback with established regression tests and an independent acceptance rule.
\end{abstract}

\section{Introduction}

Test-driven development (TDD) organizes software development around executable requirements: developers write tests that describe the intended behavior, implement the corresponding functionality, and use those tests to check subsequent changes \citep{beck2002}. This discipline is especially relevant to AI coding agents because a natural-language request permits many possible implementations, whereas an executable test gives the agent a concrete result against which to improve its code. As large language models (LLMs) take on code generation and repository repair, the connection between tests and successive implementation decisions becomes central to the reliability of the resulting software \citep{chen2021,jimenez2024,fakhoury2024}.

Existing systems already use generated tests to structure this interaction. CodeT evaluates candidate programs with generated tests, AlphaCodium uses them throughout an iterative coding process, and AgentCoder assigns programming and test design to separate agents \citep{chen2022codet,ridnik2024,huang2023agentcoder}. Guidance for Claude Code similarly describes a workflow in which one session writes tests and another writes the implementation \citep{anthropic2026bestpractices}. In an industrial integration setting, \citet{krawczyk2023} describes generative TDD based on historical input--output pairs that guide code generation and validate the generated mappings. These approaches make testing part of the development process itself, with test results shaping the next implementation.

The same feedback that guides development also allows an agent to adapt to the tested examples without satisfying the full behavioral contract. A program may implement only the exercised cases or special-case their inputs, and an agent with broader access may attempt to change the tests. Additional benchmark tests have exposed incorrect programs that passed the original evaluation \citep{liu2023}, while reward-optimization experiments have documented direct test subversion \citep{baker2025}. This is an instance of the more general discrepancy between optimizing an observed objective and satisfying its intended purpose \citep{amodei2016,skalse2022}. Increasing the number of tests and generating new tests during development address different parts of the problem, because acceptance also depends on what the coding agent has learned about those tests. A development framework must therefore specify both how new failures reach the coder and what evidence justifies accepting the resulting implementation.

\subsection{Contribution}

We propose GTDD as a general formulation of test-driven development with a separate generative tester. A human specifies the behavioral contract, the input schema, and the testing budget, while the tester creates concrete cases after the coding agent commits each candidate. The evaluator obtains expected responses from a trusted reference implementation and returns a controlled number of reduced failures to guide the next coding round. Saving generated cases and their reductions makes each discovered failure a regression test for later candidates. The formulation specifies what each agent may read and assigns independent random audits to the acceptance decision.

Our theoretical analysis bounds the probability of accepting a candidate whose failure rate exceeds a prescribed tolerance on a fixed finite population. For a candidate class that can realize arbitrary casewise losses, we obtain the exact number of visible tests required for a uniform guarantee. For a hidden test set used repeatedly, we apply the description-length argument of \citet{dwork2015generalization} to the possible feedback histories and give an explicit adaptive construction with a lower bound on its success probability. We also establish conditions under which generated tests support a sampling guarantee and derive fresh-audit bounds for adaptive candidates. The acceptance analysis distinguishes simultaneous conclusions about individual requirement families from acceptance of the repository as a whole, and extends repository acceptance to an unbounded sequence of development rounds through a summable error allocation.

We evaluate development feedback in a paired experiment on a stateful key-value store using the same language model as coder and tester. Both policies that generated new tests in each repair round ended with lower mean failure rates than the policy that generated tests once before development, which also requested fewer traces. The difference is concentrated in a few blocks, although the median difference is negative as well and the sign test detects it. Four further conditions that request equal numbers of traces produce intermediate rates with no detectable difference between adjacent conditions, so the study does not attribute the advantage to regeneration, to the number of traces, or to access to earlier counterexamples alone. Testers with and without source access showed no detectable difference, and neither dynamic policy had a detectable advantage over the property-based baseline drawn from the evaluation distribution. Applied to the recorded candidates, a fresh audit gives few expected acceptances, in line with their failure rates.

GTDD brings together practices already represented in generative TDD and multi-agent coding systems, and gives their interaction an explicit statistical account. Whereas the generative TDD of \citet{krawczyk2023} validates generated code against historical input--output pairs, GTDD generates new inputs after each candidate is committed and separates this development feedback from the evidence used for acceptance. Its contribution is the organization of candidate commitment, testing information, counterexample feedback, and acceptance within one development process. Existing tests remain useful as executable examples and regression checks, while generated cases extend the feedback available during development. This combination is particularly relevant to applications with natural-language or speech input, such as telephone agents, where scenario generation and evaluation of acceptable behavior require separate design decisions. Section~\ref{sec:practice} describes how the protocol accommodates these applications.

\subsection{Related Work}

GTDD connects automated test generation with adaptive evaluation. The relevant distinction among prior approaches is the role assigned to tests: specifying behavior, guiding implementation, or supporting an acceptance decision.

\paragraph{Test generation and counterexample-guided development.}
Property-based testing generates structured inputs from a specification \citep{claessen2000}, while directed testing and symbolic execution search program behavior more systematically \citep{godefroid2005,cadar2008}. Language models have also been used to extend search-based generation and to fuzz deep-learning libraries \citep{lemieux2023,deng2023}. Counterexample-guided inductive synthesis provides the underlying iterative structure: a verifier finds an input on which a candidate fails, and that input constrains the next candidate \citep{solarlezama2006,solarlezama2008}. Delta debugging makes such feedback easier to use by simplifying a failing input while preserving the failure \citep{zeller2002}. GTDD combines this structure with a generative tester and specifies a statistical acceptance rule when verification proceeds through sampled cases.

\paragraph{Generated tests in AI coding.}
CodeT uses agreement with generated tests to select candidate programs, and AlphaCodium incorporates those tests into iterative repair \citep{chen2022codet,ridnik2024}. Test-driven interactive generation uses tests to clarify user intent \citep{fakhoury2024}, and AgentCoder separates the programming, test-design, and test-execution roles \citep{huang2023agentcoder}. Class-level test-driven generation implements interdependent methods under method-level public tests with execution feedback and a bounded number of repairs \citep{liang2026classlevel}. CodeHacker generates adversarial inputs for a candidate and computes their expected outputs with a reference solution \citep{shi2026codehacker}, and AdverTest alternates between test generation and mutant generation to expose weaknesses in the tests \citep{chang2026advertest}. Together with the industrial use of generative TDD described by \citet{krawczyk2023}, these systems establish a substantial basis for generative testing. \citet{menet2026} give a probabilistic account of the two main uses of an execution environment, selecting among generated programs and conditioning generation on feedback, and show that the ambiguity of the task description limits what feedback-conditioned generation can achieve. Their analysis concerns how efficiently generation reaches a correct program, whereas ours concerns the evidence that sampled tests provide about a committed candidate. We organize their shared development structure around explicit candidate commitment and feedback rules, then analyze the evidence available for accepting an adaptive candidate.

\paragraph{Correctness and test-based objectives.}
Benchmarks such as SWE-bench evaluate repository repair, while SWE-bench Verified improves the alignment between issues and their evaluation tests through human review \citep{jimenez2024,chowdhury2024}. EvalPlus strengthens tests for generated code, and LiveCodeBench addresses contamination through temporally collected tasks \citep{liu2023,jain2024}. Our setting concerns a further source of dependence: the candidate learns from tests during its own development. The distinction between an observed objective and intended behavior is central to work on reward gaming \citep{amodei2016,skalse2022}, and the coding experiments of \citet{baker2025} give concrete examples of test subversion. Paired interventions that change only the expected outputs of visible tests measure how far generated programs follow the rules those tests express, and the extent varies widely across models \citep{liang2026follow}. We describe test passing, population accuracy, and protected-file modification attempts as separate observable outcomes so that a correctness error is not treated as evidence of intent.

\paragraph{Adaptive evaluation.}
Adaptive data analysis explains how repeated feedback can change the validity of conclusions based on a fixed evaluation sample \citep{dwork2015generalization,russo2016}. Our bound for repeated testing instantiates the description-length argument in Section~2.3 of the full version of \citet{dwork2015generalization}: count the possible histories and combine that count with an exact finite-population bound for each fixed candidate. Controlled feedback can also permit valid reuse, as demonstrated by the reusable holdout \citep{dwork2015reusable}. Related concerns arise in adversarial evaluation, where an adaptive attack changes the interpretation of performance on a fixed collection of adversarial examples \citep{madry2018,athalye2018,tramer2020}. In GTDD, fresh sampling after candidate commitment supplies an acceptance sample whose conditional distribution is known despite earlier development feedback.

\paragraph{Sequential inference.}
Sequential tests and time-uniform inference provide error control when the stopping time depends on observed results \citep{wald1945,ville1939,howard2021}. The literature on testing by betting and safe inference develops this principle through nonnegative evidence processes \citep{shafer2021,grunwald2024,ramdas2023}. We use these tools with an explicit hypothesis about the candidate under evaluation. A process for one fixed implementation supports sequential inspection of that implementation, whereas a product across different implementations tests a statement about the sequence. This distinction determines which evidence can support the current acceptance decision.

\section{Problem Setup}
\label{sec:setup}

We model development as an adaptive sequence of candidate implementations evaluated against a behavioral contract. A fixed population of test cases defines the accuracy target, and the information supplied to each agent determines how later candidates can depend on earlier tests.

\subsection{Evaluation Model}

Let $q\in\calQ$ index a repository issue, where $\calQ$ is the collection of issues under study. Each issue specifies a behavioral contract and an operation schema, together with the source files that may be edited and the files that are protected. We evaluate the contract through requirement families indexed by $\calJ_q$, each corresponding to a specified class of behavior.

\paragraph{Test cases.}
For family $j\in\calJ_q$, the finite population of size $N_{q,j}$ is
\begin{equation}
\calU_{q,j}=\{z_{q,j,1},\ldots,z_{q,j,N_{q,j}}\}.
\label{eq:finite-population}
\end{equation}
Each case is a sequence of operations with a specified initial state, which we call a trace. Cases within a family are distinct, although families may overlap. The populations are fixed before candidate generation, with trace-length limits and argument domains defining the evaluation target. A population can enumerate a bounded domain or contain a fixed selection from that domain. Complete evaluation refers to all listed cases, and the trusted reference implementation recomputes their expected responses whenever they are executed.

\paragraph{Candidate loss.}
For candidate $f$, let $\ell(f;z)\in\{0,1\}$ equal one when its response differs from the reference or it violates the execution contract. Violations include malformed responses or an incorrect number of responses, as well as runtime errors and resource-limit failures assigned to the candidate. A failure of the trusted evaluator or its isolation mechanism invalidates the execution. The family risk is
\begin{equation}
R_{q,j}(f)=\frac{1}{N_{q,j}}\sum_{i=1}^{N_{q,j}}\ell(f;z_{q,j,i}).
\label{eq:family-risk}
\end{equation}
Given tolerances $\rho_{q,j}\in(0,1)$, the target is satisfied when $R_{q,j}(f)\leq\rho_{q,j}$ holds for every $j\in\calJ_q$.

\paragraph{Fixed execution behavior.}
The finite-population analysis treats $\ell(f;\cdot)$ as a fixed function. The candidate therefore includes all executable code and readable data, together with dependency versions, the execution image, and the resource-policy configuration. Each case starts in the same isolated state without uncontrolled entropy or persistent writes from another case, and the execution reveals no distinction between development testing, auditing, and population evaluation. Resource limits must produce reproducible outcomes under this configuration. Repeated executions that disagree violate the fixed-loss assumption, as does a wall-clock limit whose outcome depends on machine load.

\subsection{Adaptive Development}

A trajectory consists of rounds $t=1,\ldots,T$ in which the coding agent chooses its next implementation using the information available so far. Let $\calH_t$ contain the repository history and public-test outcomes, together with earlier feedback, acceptance decisions, and any state kept by the coding session. An arbitrary measurable policy may select the candidate according to
\begin{equation}
f_t=\phi_t(\calH_t).
\label{eq:adaptive-candidate}
\end{equation}
The history includes the randomness used by the coding policy, so this representation permits randomized adaptive behavior.

After the coding session ends, the test harness records attempted changes to protected files and restores any modified protected material. It then hashes a canonical manifest of the complete candidate and its execution configuration, including file additions and deletions. This commitment, denoted by $C_t$, fixes the object used in every subsequent round-$t$ test. A later change invalidates the round. The reference implementation and private evaluation records remain outside the candidate's accessible filesystem, and population evaluation leaves no persistent state that a later coding session can inspect.

\subsection{Development Feedback}

A test-generation policy specifies which development inputs are available in a round and what information the tester may use. A fixed test set is chosen before the first round and is re-executed against later candidates, whereas a dynamic policy proposes an ordered collection $D_t$ of new traces in each round.

\paragraph{Information available to the tester.}
Let $\calI_t$ contain the issue statement and contract, the operation schema, and the public examples and previously reported counterexamples with their family labels. A source-informed tester may also read the committed candidate, and it generates
\begin{equation}
D_t=\psi_t(\calI_t,f_t,U_t),
\label{eq:development-selector}
\end{equation}
where $U_t$ is its request randomness. A source-blind tester receives $\calI_t$ without $f_t$ and has no other channel to the repository history. Both policies are excluded from audit cases, family audit outcomes, and population results.

\paragraph{Counterexamples and regression tests.}
The evaluator validates the proposals and computes reference responses. Let $V_t$ contain the valid traces absent from the existing regression tests $G_{t-1}$. The failures found in the round are
\begin{equation}
F_t=\{z\in G_{t-1}\cup V_t:\ell(f_t;z)=1\}.
\label{eq:discovered-failures}
\end{equation}
A fixed feedback rule selects at most $k$ reduced counterexamples, collected in $E_t$. Reduction records the original trace privately and preserves the designated discrepancy, although the shorter trace may lie outside the target population. A counterexample is reported at most once. The cumulative set of regression tests is
\begin{equation}
G_t=G_{t-1}\cup V_t\cup E_t,
\qquad
G_0=\varnothing,
\label{eq:regression-tests}
\end{equation}
which keeps both the original inputs and their reported reductions available in later rounds. Under a fixed policy, $V_1$ contains the valid initial tests and $V_t=\varnothing$ holds for $t>1$, while newly reported reductions can still be added to $G_t$. Invalid proposals are excluded from execution but count toward the number of proposals.

\paragraph{Feedback history and audit access.}
The reported counterexamples form the history $\tau_{t-1}=(E_1,\ldots,E_{t-1})$. This history is part of $\calH_t$, which also contains any acceptance decision observed by the coder. For independent acceptance, the harness samples an audit after recording $C_t$. For family $j$, let $S_{t,j}$ be a uniform sample without replacement from $\calU_{q,j}$. Writing $\1[\cdot]$ for the indicator function, define
\begin{equation}
A_{t,j}=\1\Bigl[\sum_{z\in S_{t,j}}\ell(f_t;z)=0\Bigr].
\label{eq:audit-indicator}
\end{equation}
Audit cases and their individual outcomes are excluded from $D_t$, $E_t$, and $G_t$.

\subsection{Evaluation Criteria}

Development feedback is useful when repairs improve behavior beyond the reported inputs. We formalize this distinction through test passing and population risk, and use population risk to define an incorrect acceptance decision.

\paragraph{Test passing and target accuracy.}
Let $B_t$ be the tests whose pass criterion the coding agent observes, including public tests and any visible development tests. Define
\begin{equation}
P_t=\1\bigl[\ell(f_t;z)=0\ \text{for every } z\in B_t\bigr],
\qquad
Q_t=\1\bigl[R_{q,j}(f_t)\leq\rho_{q,j}\ \text{for every } j\in\calJ_q\bigr].
\label{eq:pass-success}
\end{equation}

\begin{definition}[Discrepancy between test passing and target accuracy]
The signed difference is $\Gamma_t=P_t-Q_t$, and the event of passing the observed tests without meeting the target is $C^{\mathrm{dev}}_t=P_t(1-Q_t)$. An average of $\Gamma_t$ can cancel positive and negative differences, whereas the frequency of $C^{\mathrm{dev}}_t$ measures the one-sided discrepancy directly.
\label{def:pass-risk-discrepancy}
\end{definition}

False acceptance occurs when the acceptance rule accepts a candidate with $Q_t=0$. Our guarantees bound the probability of this event over the testing procedure, rather than the conditional proportion of inaccurate candidates among those accepted. These quantities describe observable correctness outcomes. Intentional interference is recorded separately through attempted edits to protected files.

\paragraph{Improvement on inputs not yet reported.}
Let $D^{\mathrm{seen}}_{t,j}$ contain population cases whose identities or complete inputs have reached the coder by the end of round $t$, including public cases and cases identifiable from feedback. Write $\calU^{\mathrm{unseen}}_{t,j}=\calU_{q,j}\setminus D^{\mathrm{seen}}_{t,j}$. When this set is nonempty, define
\begin{equation}
R^{\mathrm{unseen}}_{t,j}(f)=\frac{1}{|\calU^{\mathrm{unseen}}_{t,j}|}\sum_{z\in\calU^{\mathrm{unseen}}_{t,j}}\ell(f;z).
\label{eq:unseen-risk}
\end{equation}
Evaluating $f_t$ and $f_{t+1}$ on this same set measures whether the update improves behavior beyond known inputs. A reported failure can be repaired while this risk stays unchanged or increases. The full family risk can obscure that distinction because repairing one known case lowers the mean by $1/N_{q,j}$ even when no other outcome changes. Equation~\eqref{eq:unseen-risk} is undefined when every population case has been exposed.

\section{Generative Test-Driven Development}
\label{sec:gtdd}

GTDD places a separate testing agent inside the development cycle so that each committed implementation faces newly generated inputs. The contract and evaluation rules stay fixed, while the tester uses permitted feedback to search for failures that inform the next coding round. The protocol specifies this interaction independently of the rule used to accept an implementation.

\subsection{Round Protocol}

Figure~\ref{fig:protocol} shows the order of one round. The coding session edits the repository, after which the harness restores protected material and commits the candidate and its execution configuration. Public tests and all subsequent evaluations run against that same snapshot.

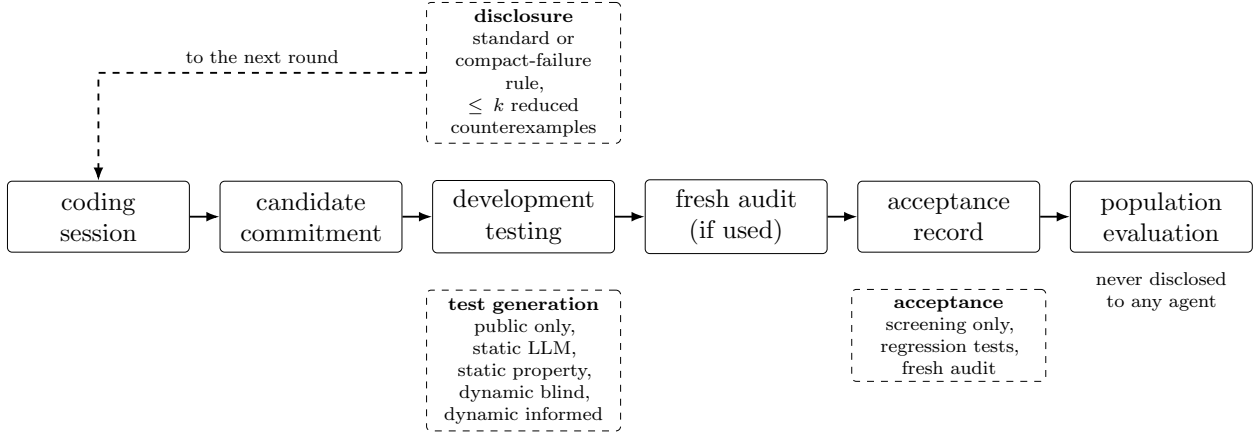
\begin{figure}[htbp]
\centering
\resizebox{\textwidth}{!}{%
\begin{tikzpicture}[
 stage/.style={draw, rounded corners=2pt, minimum height=2.6em, text width=6.0em, align=center, font=\small},
 axis/.style={draw, dashed, rounded corners=2pt, text width=6.6em, align=center, font=\scriptsize, inner sep=3pt},
 flow/.style={-{Latex[length=2mm]}, thick},
 feed/.style={-{Latex[length=2mm]}, thick, dashed},
 node distance=0.55cm and 0.45cm
]
\node[stage] (coding) {coding\\session};
\node[stage, right=of coding] (commit) {candidate\\commitment};
\node[stage, right=of commit] (testing) {development\\testing};
\node[stage, right=of testing] (audit) {fresh audit\\(if used)};
\node[stage, right=of audit] (accept) {acceptance\\record};
\node[stage, right=of accept] (population) {population\\evaluation};
\draw[flow] (coding) -- (commit);
\draw[flow] (commit) -- (testing);
\draw[flow] (testing) -- (audit);
\draw[flow] (audit) -- (accept);
\draw[flow] (accept) -- (population);
\node[axis, below=of testing] (prov) {\textbf{test generation}\\public only, static LLM,\\static property,\\dynamic blind,\\dynamic informed};
\node[axis, below=of accept] (acc) {\textbf{acceptance}\\screening only,\\regression tests,\\fresh audit};
\node[axis, above=of testing] (disc) {\textbf{disclosure}\\standard or\\compact-failure\\rule,\\$\leq k$ reduced\\counterexamples};
\draw[feed] (disc.west) -| ($(coding.north)+(0,0)$) node[pos=0.25, above, font=\scriptsize] {to the next round};
\node[font=\scriptsize, below=0.15cm of population, text width=6.2em, align=center] {never disclosed\\to any agent};
\end{tikzpicture}}
\caption{Order of a development round. The test-generation policy supplies inputs, the feedback rule returns at most $k$ reduced failures to the next coding round, and the acceptance rule determines whether the candidate is accepted. Population evaluation follows the recorded acceptance decision and is never reported to either agent.}
\label{fig:protocol}
\end{figure}

A source-informed testing session receives a read-only copy matching the commitment, together with the contract and the permitted history. A source-blind session receives the same information except for the source. Each request asks for a specified number of stateful traces without expected outputs, which the trusted reference implementation computes. When a condition uses fixed tests, their manifest is checked before they are re-executed.

The evaluator executes each valid nonduplicate proposal against both the reference and the candidate. For a failure, it first finds the shortest failing prefix and then deletes operations while preserving validity and the designated discrepancy. An output mismatch keeps the identity of the divergent operation even if its position changes, and an execution failure keeps its error class. Each shorter trace is verified before it replaces the preceding trace, with the operation limit bounding the reduction cost as described in Appendix~\ref{app:disclosure}.

At most $k$ reduced failures are returned to the coder. Each record contains the trace and its family label, the expected and observed responses, and the execution status, while the link to the original trace remains private. Passing generated inputs also become regression tests. The harness records the acceptance decision before evaluating the complete population, and that evaluation supplies no feedback to coding, test generation, or stopping.

\subsection{Development Policies}

A development condition specifies how tests are generated and which failures are returned. These choices determine the information available for the next repair while leaving the acceptance rule separately configurable.

\paragraph{Test generation.}
The protocol supports public examples alone, fixed tests generated from the specification, and tests generated in each round. Property-based generators and LLMs can supply either fixed or newly generated inputs. In a visible condition, the tests and expected responses are protected files inside the repository, whereas a hidden condition keeps them outside the coder's access. Dynamic LLM testing additionally specifies whether the tester may inspect the source. We call the tester adversarial because its objective is to propose valid inputs on which the committed candidate disagrees with the reference implementation, using the information in $\calI_t$, the source when the condition permits it, and a fixed number of traces per request. Its success is the number of such failures it finds, and it has no access to audit samples or population results, so adversarial testing shapes development feedback but not the evidence for acceptance.

\paragraph{Feedback selection.}
The standard rule prioritizes diversity across requirement families and failure mechanisms before selecting shorter traces. The compact-failure rule prioritizes short traces with few distinct operations and then families already reported. This second rule concentrates feedback on compact failures and allows a controlled comparison of how counterexample selection affects repair. Both rules act on the same observed development failures without access to audit allocations or population outcomes. Appendix~\ref{app:disclosure} specifies their deterministic ordering and reduction budgets.

\paragraph{Development outcomes.}
Discovery counts valid nonduplicate failures on the current candidate, with invalid and duplicate proposals kept in the denominator. Repair measures the next candidate's performance on the reported counterexamples, while equation~\eqref{eq:unseen-risk} measures improvement on the same previously unknown population cases. A regression occurs when a later candidate fails a saved test that it previously passed. Reporting these outcomes separately reveals whether a change in the final risk comes from discovering different failures, repairing the reported inputs, or preserving earlier behavior. Appendix~\ref{app:elementary} states elementary properties of these outcomes and of the information available to the tester.

A comparison of development policies keeps the coding budget and acceptance rule fixed. The generation budget and repeated execution costs also belong to the specification of each policy, since an initial collection of tests incurs both generation cost and re-execution cost in later rounds.

\subsection{Acceptance Rules}
\label{sec:acceptance-comparison}

The acceptance rule states which evidence is sufficient to accept the committed candidate. Every rule requires an intact snapshot and passing public tests. Public screening stops at that criterion, and a regression-based rule additionally requires every saved development test to pass.

GTDD uses a fresh audit of every requirement family to control false acceptance. A condition may also require the regression tests to pass, and this additional requirement can only reduce the probability of accepting an inaccurate candidate. With a positive risk tolerance, an audit-only rule can accept a candidate that fails a known development case, so the choice to require regression tests is part of the acceptance specification. Comparisons between acceptance rules keep this choice and the development feedback fixed.

The coder receives only the aggregate acceptance decision. Sampled inputs and family-specific results remain private. A trajectory can follow a prescribed stopping policy, including stopping at its first acceptance or continuing for a specified number of rounds. The upper bounds in Section~\ref{sec:theory} control the event that at least one accepted candidate exceeds a family tolerance during the evaluated trajectory.

\subsection{Integration with Existing Tests}

Existing tests provide the starting point for GTDD by expressing the contract through concrete examples and protecting behavior that must remain correct. Re-executing them gives every candidate an inexpensive initial check, while generated tests extend the range of failures available for development feedback.

The combined procedure keeps the initial tests, adds valid generated traces after each commitment, and saves the counterexamples reported to the coder. By equation~\eqref{eq:regression-tests}, a failure found through generation becomes a fixed regression test for subsequent candidates. A strong property-based generator can also supply the initial tests when the input schema captures the relevant behavior. In the experiment of Section~\ref{sec:case-study}, the property-based baseline and the dynamic policies had no detectable difference in final failure rate, which makes this baseline an important part of the comparison.

Acceptance draws a fresh sample after commitment from the population fixed before development. The sample size is determined before the draw, and individual audit cases remain outside the development process. This division gives examples and regression tests a stable role while assigning error control over the population to a known sampling rule.

\subsection{Applications}
\label{sec:practice}

The balance between existing tests and generated tests depends on how inputs can be produced and how correct behavior can be evaluated. GTDD makes both choices explicit, whether inputs are structured operation traces or natural-language interactions.

\paragraph{Structured inputs.}
For a key-value store, parser, or formally specified interface, an operation schema describes valid inputs and a reference implementation supplies expected outputs. A bounded collection of traces then defines a finite population on which committed candidates have fixed losses. This setting directly supports the post-commit sampling assumption used in the acceptance analysis.

\paragraph{Language and speech.}
A telephone agent encounters dialogue patterns such as interruptions and topic changes, while the consequences of an utterance depend on the agent's dialogue state and tool calls. Generated scenarios can vary these patterns and target the flows exposed by the candidate implementation. For speech input, scenarios can also represent recognition errors, accents, and background noise. Existing scripted calls remain useful for requirements such as identity verification before an account change and transfer to a person on request. Generated dialogues complement these checks and become regression tests when they expose a failure.

\paragraph{Evaluation criteria.}
Natural-language responses often admit several acceptable outcomes, so evaluation may use a rubric, a model-based judge, or human review instead of a reference implementation. The finite-population guarantees concern the loss assigned by the specified evaluator on the listed scenarios and require reproducible candidate and evaluator behavior. The validity of that criterion is a separate question about its agreement with the intended behavior. Fixing the evaluation procedure before development and keeping audit scenarios private preserves the distinction between generating useful feedback and deciding whether the candidate meets the stated target.

\section{Theoretical Analysis}
\label{sec:theory}

We seek a frequentist bound on the probability of accepting an implementation whose failure rate exceeds the prescribed tolerance. The population and each committed candidate's losses are fixed, while probability is taken over test sampling and any independent randomization in development. The analysis first relates false acceptance to sample size, then accounts for the information available when the candidate is chosen. This establishes which guarantees follow from known tests, repeated feedback, and newly sampled tests, and yields a sequential acceptance rule for GTDD.

Fix an issue and one requirement family, write $\calU=\{z_1,\ldots,z_N\}$ for its population of size $N$, and let $R(f)$ be the family risk in equation~\eqref{eq:family-risk} with tolerance $\rho\in(0,1)$. We call a candidate bad when $R(f)>\rho$ holds, and let $\alpha\in(0,1)$ be the target error probability. Sample sizes are integers between zero and $N$. The lower-bound constructions use an unrestricted casewise candidate class, which can realize every binary loss pattern on $\calU$ and includes candidates that recognize complete inputs and return their known specified outputs. For the upper bounds, candidate selection may be arbitrarily adaptive under the stated information conditions.

\subsection{False Acceptance Probability}

For a candidate fixed before sampling, false acceptance occurs when the sampled inputs miss all its failures. Counting those inputs gives an exact finite-population probability, which supplies the basic bound used throughout the analysis.

\begin{proposition}[Exact acceptance probability]\label{prop:audit-probability}
Let $K(f)=\sum_{i=1}^{N}\ell(f;z_i)$, and let $S$ be a uniform sample of $n$ cases drawn without replacement independently of the fixed candidate $f$. Then, we have
\begin{equation}
\Pr(f\text{ passes }S)=\frac{\binom{N-K(f)}{n}}{\binom{N}{n}}\leq\Bigl(1-\frac{K(f)}{N}\Bigr)^{n},
\label{eq:hypergeometric}
\end{equation}
with $\binom{a}{n}=0$ for $n>a$. Define $K_\rho=\lfloor\rho N\rfloor+1$ and $q_{N,\rho}(n)=\binom{N-K_\rho}{n}/\binom{N}{n}$. If $R(f)>\rho$ holds, then we have $\Pr(f\text{ passes }S)\leq q_{N,\rho}(n)\leq(1-\rho)^{n}$.
\end{proposition}

\begin{proof}
The candidate passes exactly when all sampled indices lie among its $N-K(f)$ passing cases, which gives the equality. If $n\leq N-K(f)$ holds, each factor of $\prod_{r=0}^{n-1}(N-K(f)-r)/(N-r)$ is at most $1-K(f)/N$. For larger $n$, the probability is zero. The expression is nonincreasing in $K(f)$, so the largest acceptance probability among bad candidates occurs at $K_\rho$, and $K_\rho/N>\rho$ gives the final inequality.
\end{proof}

The integer threshold $K_\rho$ is the smallest number of failures above tolerance. Using it keeps the exact endpoint: once $n\geq N-K_\rho+1$ holds, every passing candidate must be within tolerance.

A zero-failure audit can reject a candidate within tolerance that fails only a few cases, so an audit may instead accept when the sample contains at most $c$ failures. The same argument then applies to the lower tail of the hypergeometric count.

\begin{proposition}[Audits that permit failures]\label{prop:audit-failures}
Let $c$ be a nonnegative integer, let $S$ be as in Proposition~\ref{prop:audit-probability}, and let the audit accept when $S$ contains at most $c$ failing cases. Then, the acceptance probability of the fixed candidate $f$ is
\begin{equation}
\Pr\bigl(|\{z\in S:\ell(f;z)=1\}|\leq c\bigr)=\frac{1}{\binom{N}{n}}\sum_{h=0}^{\min\{c,n\}}\binom{K(f)}{h}\binom{N-K(f)}{n-h}.
\label{eq:audit-failures}
\end{equation}
It is nonincreasing in $K(f)$, so every candidate with $R(f)>\rho$ is accepted with probability at most the value $q^{(c)}_{N,\rho}(n)$ of equation~\eqref{eq:audit-failures} at $K(f)=K_\rho$, and $q^{(0)}_{N,\rho}=q_{N,\rho}$.
\end{proposition}

\begin{proof}
The number of sampled failures is hypergeometric with parameters $(N,K(f),n)$, which gives equation~\eqref{eq:audit-failures}. Adding one failing case to a candidate can only increase the number of sampled failures for every sample, so the probability of at most $c$ sampled failures cannot increase.
\end{proof}

The fresh-audit bounds of Section~\ref{sec:fresh-theory} use only this conditional bound for each fixed candidate, so they hold with $q^{(c)}_{N,\rho}$ in place of $q_{N,\rho}$. The remaining results are stated for the zero-failure rule.

\subsection{Sample Size for Known Tests}
\label{sec:static-theory}

When a candidate is selected after the tests and their expected responses are revealed, it can adapt directly to the tested inputs. In the unrestricted casewise class, the largest accepted risk is then determined by the fraction of the population left untested.

\begin{proposition}[Adaptation to known tests]
In the unrestricted casewise class, if $S\subsetneq\calU$ holds and its required outputs are known, there is a candidate that passes every case in $S$ and fails every case outside it.
\label{prop:known-tests}
\end{proposition}

\begin{proof}
The class contains the candidate that returns the required output on $S$ and an incorrect output on its complement. Its risk is $1-|S|/N$.
\end{proof}

\begin{proposition}[Exact size requirement for known tests]
If $|S|=n$ holds and the rule accepts every candidate that passes $S$, the largest accepted risk in the unrestricted casewise class is $1-n/N$. A uniform guarantee that every accepted candidate has risk at most $\rho$ holds exactly when
\begin{equation}
n\geq N-\lfloor\rho N\rfloor.
\label{eq:static-necessity}
\end{equation}
\label{prop:static-size}
\end{proposition}

\begin{proof}
A passing candidate fails at most $N-n$ cases, and Proposition~\ref{prop:known-tests} attains that count. Because $N-n$ is an integer, $(N-n)/N\leq\rho$ is equivalent to equation~\eqref{eq:static-necessity}.
\end{proof}

For $N=500$ and $\rho=0.05$, the threshold is $475$ cases. With fewer known tests, the candidate of Proposition~\ref{prop:known-tests} passes every observed test while missing the target, which is the event $C^{\mathrm{dev}}_t$ of Definition~\ref{def:pass-risk-discrepancy}. The result depends on the candidate class. If a finite class $\calF$ is fixed independently of a uniformly sampled test set, a union bound gives $\Pr(\exists f\in\calF: R(f)>\rho,\ f\text{ passes }S)\leq|\calF|\,q_{N,\rho}(n)$, even if the tests are revealed before a candidate is selected. Structural constraints on admissible implementations therefore change the size requirement for known tests.

\subsection{Error Bounds under Repeated Feedback}
\label{sec:reused-theory}

A hidden test set initially prevents direct adaptation to its inputs, but repeated feedback reveals information about that same set. We quantify this dependence by counting the possible feedback histories. The sufficient bound applies the description-length argument of \citet[Section~2.3 of the full version]{dwork2015generalization} with the hypergeometric probability from Proposition~\ref{prop:audit-probability}. An explicit probing construction then shows how revealed cases can be used to avoid the hidden tests.

\paragraph{A sufficient bound from feedback histories.}
The relevant history contains every observation about the test set that reaches the coding agent, including an acceptance decision when it supplies additional information.

\begin{assumption}[Repeated evaluation on a hidden test set]
A uniform size-$n$ subset $S\subseteq\calU$ is sampled before the first round. Conditional on randomness $\xi$ independent of $S$, candidate $f_t$ is a function of the preceding messages $W_1,\ldots,W_{t-1}$. Each message includes every observation that depends on $S$ and is available to the coding agent in that round. For every preceding history and fixed $\xi$, message $W_s$ takes at most $L_s$ values, where $L_s$ is deterministic. Round $t$ accepts when $f_t$ passes $S$.
\label{ass:reused}
\end{assumption}

\begin{theorem}[Finite-population bound from feedback histories]
Under Assumption~\ref{ass:reused}, let $\calT_{t-1}(\xi)$ be the set of possible feedback histories before round $t$ given $\xi$. Then, we have
\begin{equation}
\begin{gathered}
\Pr\bigl(\exists\,t\in\{1,\ldots,T\}:\ f_t\text{ passes }S,\ R(f_t)>\rho\bigr)
\leq\min\Bigl\{1,\ q_{N,\rho}(n)\,\E\Bigl[\sum_{t=1}^{T}|\calT_{t-1}(\xi)|\Bigr]\Bigr\},\\
|\calT_{t-1}(\xi)|\leq\prod_{s=1}^{t-1}L_s.
\end{gathered}
\label{eq:history-bound}
\end{equation}
For raw reported subsets of size at most $k$ together with an acceptance bit, one may take $L_s=2\sum_{i=0}^{\min\{k,N\}}\binom{N}{i}\leq2(N+1)^{k}$. A sufficient exponential condition is
\begin{equation}
n\geq\frac{\log(T/\alpha)+(T-1)\bigl(k\log(N+1)+\log2\bigr)}{\log\bigl(1/(1-\rho)\bigr)}.
\label{eq:history-count}
\end{equation}
If the acceptance bit is unobserved or determined by the other feedback, the factor $2$ and the term $\log2$ are omitted.
\label{thm:feedback-history}
\end{theorem}

The exact bound uses both the number of histories and the finite-population endpoint. For $N=500$, $k=6$, $T=4$, and $\rho=\alpha=0.05$, equation~\eqref{eq:history-count} gives $2{,}308$ cases with the acceptance bit and $2{,}267$ without it. Equation~\eqref{eq:history-bound}, by contrast, is satisfied at $474$ and $473$ cases, respectively, while $475$ passing cases give a deterministic guarantee. Exact counting remains useful when the exponential approximation exceeds the population size.

The message count must reflect the information actually returned. For raw reported cases, the fixed candidate and the case determine the reference and candidate responses. Additional observations such as timing or a test score enlarge the possible messages. A deterministic reduction based only on each reported case and the candidate cannot increase that count, although the reduced input may still identify the original case. Appendix~\ref{app:proofs} gives the proof.

\paragraph{A lower bound from adaptive probing.}
The following construction uses disjoint regions of the input population. During the probing rounds, the candidate fails in one region so that feedback identifies its sampled members. The last candidate then fails on the parts of those regions that remain unreported.

\begin{theorem}[Adaptive probing lower bound]
Suppose the candidate class is unrestricted casewise, the agent knows $\calU$ and its specified outputs, and the feedback rule returns every previously unreported failing case of $S$ whenever there are at most $k$ of them. Let $r=T-1$, and let integer $m$ satisfy $rm\leq N$ and $r(m-k)>\rho N$. If feedback continues for all $T$ prescribed rounds, there is a policy whose round-$T$ candidate is accepted with risk above $\rho$ with probability at least $\max\{0,1-r\,p(m)\}$, where $p(m)=\Pr\bigl(\mathrm{Hyp}(N,m,n)>k\bigr)$ is the upper tail of the hypergeometric count of sampled cases in a region of size $m$.
\label{thm:probe}
\end{theorem}

\begin{corollary}[A sufficient choice of probe size]
Let $T\geq2$, $k\geq1$, and $n\geq1$, and set $r=T-1$ and $m=\lfloor kN/(2n)\rfloor$. If $rm\leq N$ and $r(m-k)>\rho N$ hold, Theorem~\ref{thm:probe} applies with a hypergeometric count whose mean is $nm/N\leq k/2$. Ignoring integer rounding, the risk condition is equivalent to
\begin{equation}
n<n^{*}:=\Bigl(\frac{2\rho}{rk}+\frac{2}{N}\Bigr)^{-1},
\label{eq:nstar}
\end{equation}
and $n\geq rk/2$ suffices for the disjoint regions to fit in the population.
\label{cor:nstar}
\end{corollary}

This choice of $m$ yields $n^{*}\approx104.65$ when $N=500$, $k=6$, $T=4$, and $\rho=0.05$. After rounding, it meets the risk condition through $n\leq100$. Choosing $m=15$ directly also permits larger values of $n$, with the resulting probabilities calculated in Section~\ref{sec:numerics}. The sufficient bound covers all policies satisfying the message-count assumption, while the lower bound gives the success probability of the specified probing policy.

\begin{remark}[Case identity after reduction]
Reducing $(\mathrm{query}(i),\mathrm{noop})$ to $(\mathrm{query}(i))$ preserves the identifying value $i$. A shorter input can therefore reveal its original population case. The information relevant to Theorem~\ref{thm:feedback-history} depends on the mapping from original cases to feedback, including how many cases can yield the same message, rather than on the fraction of traces that remain unchanged.
\label{rem:reduction}
\end{remark}

\subsection{Error Bounds for Generated Tests}
\label{sec:dynamic-theory}

Generating new tests after a candidate is fixed changes the acceptance analysis when their sampling distribution is known. Uniformly sampled batches give a direct bound on the current candidate, whereas targeted generation requires a relation between its selection distribution and the target population.

\begin{proposition}[Uniform generation in each round]
Suppose each round draws a batch $S^{\mathrm{new}}_t$ of $m$ cases uniformly without replacement after commitment, conditionally independently of the candidate given the previous history. Repetition across rounds is allowed, and $V_t=S^{\mathrm{new}}_t\setminus G_{t-1}$ contains the new cases. If acceptance requires all cumulative regression tests to pass, then for every round we have
\begin{equation}
\Pr\bigl(f_t\text{ passes }G_t,\ R(f_t)>\rho\mid\calH_t,f_t\bigr)\leq q_{N,\rho}(m)\leq(1-\rho)^{m},
\label{eq:dynamic-bound}
\end{equation}
and the probability of any false acceptance over $T$ rounds is at most $T\,q_{N,\rho}(m)$.
\label{prop:dynamic}
\end{proposition}

\begin{proof}
Passing $G_t$ implies passing the new batch because $S^{\mathrm{new}}_t\subseteq G_{t-1}\cup V_t\subseteq G_t$. Proposition~\ref{prop:audit-probability} applies conditionally in each round, and a union bound gives the result over $T$ rounds. Requiring additional tests to pass can only reduce acceptance.
\end{proof}

An overall error budget $\alpha$ over $T$ rounds requires $q_{N,\rho}(m)\leq\alpha/T$. For $N=500$, $\rho=0.05$, and $m=12$, the exact single-round bound is approximately $0.523$, so this batch size yields a weak acceptance criterion. The role of the distribution becomes explicit in the next result.

\begin{proposition}[Candidate-dependent test selection]
Let $Q_f$ be the distribution used to select an input after observing candidate $f$. Without a specified relation between $Q_f$ and the uniform distribution on $\calU$, the expected selected-test loss $\E_{Z\sim Q_f}[\ell(f;Z)]$ does not identify $R(f)$.
\label{prop:selection}
\end{proposition}

\begin{proof}
A permissible selector can concentrate on a strict subset of $\calU$. Candidates that agree on that subset and differ on its complement have the same selected-test loss and different population risks.
\end{proof}

Targeting suspected defects can increase the number of failures found, which is the development objective of an informed tester. To infer population risk from such tests, a coverage condition is needed. For example, if a known $\kappa\in(0,1]$ satisfies $Q_f(z)\geq\kappa/N$ for every $z\in\calU$, then we have $\E_{Z\sim Q_f}[\ell(f;Z)]\geq\kappa\,R(f)$. GTDD measures targeted generation through the failures and repairs it produces, and uses an explicit random sampling rule for acceptance.

\subsection{Sequential Error Control}
\label{sec:fresh-theory}

Fresh audits permit valid acceptance decisions even when earlier tests have influenced candidate selection. We first control errors over a fixed collection of trajectories and then extend that control to an unbounded number of rounds. The bound for accepting a complete repository is sharper than the bound for reporting simultaneous conclusions about every requirement family.

Let $a=1,\ldots,M$ index trajectories with at most $T$ rounds and $J$ families. Use $A_{a,t,j}$ for the family audit indicator and $Q_{a,t}$ for target accuracy, extending the notation of Section~\ref{sec:setup}. Let $N_{a,j}$ be the family population size and $\rho_{a,t,j}$ the tolerance used in $Q_{a,t}$. In GTDD, this is the issue-family tolerance fixed before development. Denote by $B_{a,t,j}$ the event that family $j$ exceeds its tolerance but passes its audit. The sigma-field $\calF^{-}_{a,t}$ contains all information available to the trusted harness just before sampling, including the committed candidate.

\begin{assumption}[Sampling after candidate commitment]
Conditional on $\calF^{-}_{a,t}$, sample $S_{a,t,j}$ is uniform without replacement from its fixed population. Its size $n_{a,t,j}$ is chosen before the draw, and the candidate's loss function and tolerance $\rho_{a,t,j}$ are fixed. Dependence between samples from different families is allowed.
\label{ass:post-commit}
\end{assumption}

\paragraph{Simultaneous family conclusions.}
The probability that any above-tolerance family passes its audit is bounded by summing its conditional error probabilities across families, rounds, and trajectories.

\begin{theorem}[Sequential family-wise error control]
Under Assumption~\ref{ass:post-commit}, we have
\begin{equation}
\Pr\Bigl(\bigcup_{a=1}^{M}\bigcup_{t=1}^{T}\bigcup_{j=1}^{J}B_{a,t,j}\Bigr)
\leq\E\Bigl[\sum_{a=1}^{M}\sum_{t=1}^{T}\sum_{j=1}^{J}q_{N_{a,j},\rho_{a,t,j}}(n_{a,t,j})\Bigr].
\label{eq:sequential-bound}
\end{equation}
For fixed sizes and tolerances, the expectation is unnecessary, and each term is at most $(1-\rho_{a,t,j})^{n_{a,t,j}}$.
\label{thm:familywise}
\end{theorem}

\begin{proof}
The committed loss function determines whether a family exceeds its tolerance. Conditional on $\calF^{-}_{a,t}$, Proposition~\ref{prop:audit-probability} bounds the probability that such a family passes by $q_{N_{a,j},\rho_{a,t,j}}(n_{a,t,j})$. Taking expectations and a union bound proves equation~\eqref{eq:sequential-bound}. Rounds excluded by a stopping rule contribute zero event indicators.
\end{proof}

\begin{corollary}[Family-wise sample size]
For common $N$, $\rho$, and $n$, the condition $q_{N,\rho}(n)\leq\alpha/(MTJ)$ suffices. A sufficient exponential condition is
\begin{equation}
n\geq\frac{\log(MTJ/\alpha)}{\log\bigl(1/(1-\rho)\bigr)}.
\label{eq:sample-size}
\end{equation}
For $N=500$, $M=8$, $T=4$, $J=5$, and $\rho=\alpha=0.05$, the smallest exact size is $131$, whereas equation~\eqref{eq:sample-size} gives $158$.
\label{cor:sample-size}
\end{corollary}

\paragraph{Acceptance of the complete repository.}
If acceptance requires all family audits to pass, any bad candidate must pass the audit of at least one family known before sampling to be above tolerance. This observation replaces the sum over families by their largest conditional bound. Let $\mathrm{Accept}_{a,t}$ denote the indicator that the repository is accepted.

\begin{theorem}[Joint repository acceptance]
Suppose acceptance requires $A_{a,t,j}=1$ for every family, possibly together with passing public and saved development tests. Under Assumption~\ref{ass:post-commit}, we have
\begin{equation}
\begin{split}
&\Pr\bigl(\exists\,a\in\{1,\ldots,M\},\ t\in\{1,\ldots,T\}:\ \mathrm{Accept}_{a,t}=1,\ Q_{a,t}=0\bigr)\\
&\qquad\leq\E\Bigl[\sum_{a=1}^{M}\sum_{t=1}^{T}\max_{1\leq j\leq J}q_{N_{a,j},\rho_{a,t,j}}(n_{a,t,j})\Bigr].
\end{split}
\label{eq:joint-bound}
\end{equation}
\label{thm:joint}
\end{theorem}

\begin{proof}
Condition on $\calF^{-}_{a,t}$. For a bad candidate, let $j^{*}$ be the lowest-indexed family whose risk exceeds its tolerance. This choice depends on the fixed candidate and population, before the new samples are drawn. Repository acceptance implies that family $j^{*}$ passes, with conditional probability at most $q_{N_{a,j^{*}},\rho_{a,t,j^{*}}}(n_{a,t,j^{*}})$, which is bounded by the displayed maximum. For a candidate within tolerance in every family, the false-acceptance event is empty. Taking expectations and a union bound over trajectories and rounds proves the result.
\end{proof}

With common parameters, joint acceptance requires $q_{N,\rho}(n)\leq\alpha/(MT)$, so the preceding setting needs $108$ cases per family. The number of families still determines the execution cost of $Jn$ cases per round. The distinction between this event for the whole repository and simultaneous family conclusions accounts for the different sample sizes.

\begin{corollary}[Comparison with known tests]
In the unrestricted casewise class, a known fixed test set needs $N-\lfloor\rho N\rfloor$ cases to exclude every accepted candidate above tolerance. For one fixed candidate and one family, a uniform audit after commitment requires $q_{N,\rho}(n)\leq\alpha$. The exponential count $\log(1/\alpha)/\log(1/(1-\rho))$ suffices whenever it is feasible, and the finite-population endpoint always suffices. At $N=500$ and $\rho=\alpha=0.05$, the smallest exact audit size is $54$, compared with $475$ known tests.
\label{cor:scale-vs-fresh}
\end{corollary}

Theorems~\ref{thm:familywise} and~\ref{thm:joint} extend the single-candidate comparison to a prescribed horizon and allow arbitrary stopping within that horizon.

\paragraph{An unbounded sequence of candidates.}
\label{sec:eprocess}
A summable allocation of error probabilities controls acceptance even when the total number of development rounds is not fixed in advance. Each allocation applies to the current candidate's new audit.

\begin{theorem}[Error allocation over an unbounded horizon]
Suppose Assumption~\ref{ass:post-commit} holds in every audited round. For $a\in\{1,\ldots,M\}$ and $t\geq1$, choose a nonnegative allocation $\varepsilon_{a,t}$ before sampling, with $\sum_{a=1}^{M}\sum_{t=1}^{\infty}\varepsilon_{a,t}\leq\alpha$ almost surely. Choose sample sizes satisfying $\max_{1\leq j\leq J}q_{N_{a,j},\rho_{a,t,j}}(n_{a,t,j})\leq\varepsilon_{a,t}$. If repository acceptance requires every family audit to pass, then for arbitrary adaptive candidate selection and stopping we have
\begin{equation}
\Pr\bigl(\exists\,a\in\{1,\ldots,M\},\ t\geq1:\ \mathrm{Accept}_{a,t}=1,\ Q_{a,t}=0\bigr)\leq\alpha.
\label{eq:allocation}
\end{equation}
\label{thm:allocation}
\end{theorem}

\begin{proof}
The proof of Theorem~\ref{thm:joint} bounds the conditional false-acceptance probability in each reached round by $\varepsilon_{a,t}$. The indicator of any false acceptance is at most the sum of the round indicators. Taking expectations and applying monotone convergence to the nonnegative infinite sum bounds the probability by $\E[\sum_{a=1}^{M}\sum_{t=1}^{\infty}\varepsilon_{a,t}]\leq\alpha$.
\end{proof}

For $M$ trajectories, $\varepsilon_{a,t}=\alpha/(Mt(t+1))$ has total mass $\alpha$. With $M$ and the tolerances fixed, the exponential sufficient sample count grows logarithmically as $t$ increases until the exact finite-population endpoint is smaller. The candidate may change arbitrarily between rounds. Allocating error to individual triples $(a,t,j)$ instead gives simultaneous family conclusions through Theorem~\ref{thm:familywise}.

\paragraph{Sequential testing within one candidate.}
For a fixed snapshot, a sequential test can inspect the audit one case at a time. Let $I_r$ indicate that the first $r$ draws of a uniformly random ordering of $\calU$ contain no failures, and define $W_r=I_r/q_{N,\rho}(r)$ for $0\leq r\leq N-K_\rho$. Under the null $K(f)\geq K_\rho$, the process $(W_r)$ is a nonnegative supermartingale starting at one. Given $r<N-K_\rho$ failure-free draws, the next draw passes with probability $(N-K(f)-r)/(N-r)$, which is at most $(N-K_\rho-r)/(N-r)=q_{N,\rho}(r+1)/q_{N,\rho}(r)$. Ville's inequality bounds by $\varepsilon$ the probability of ever reaching $W_r\geq1/\varepsilon$, for $\varepsilon\in(0,1)$ \citep{ville1939,ramdas2023}. At the endpoint, $N-K_\rho+1$ passing draws establish that the candidate is within tolerance directly. A new candidate starts its own process using the allocation in Theorem~\ref{thm:allocation}.

\paragraph{Evidence across different candidates.}
A product of audit factors from different implementations addresses a different hypothesis. For one family and audits with $q_{N,\rho}(n_t)>0$, let $A_t$ be the audit indicator for candidate $f_t$, and define
\begin{equation}
\eta_t=\frac{A_t}{q_{N,\rho}(n_t)},
\qquad
Z_t=\prod_{s=1}^{t}\eta_s,
\qquad
Z_0=1.
\label{eq:evalue}
\end{equation}

\begin{theorem}[E-process for a sequence of bad candidates]
Let $H_0$ be the hypothesis that every selected candidate satisfies $R(f_t)>\rho$ almost surely. Under $H_0$ and post-commit sampling, $(Z_t)$ is a nonnegative supermartingale, and we have
\begin{equation}
\Pr_{H_0}\bigl(\exists\,t\geq1:\ Z_t\geq1/\alpha\bigr)\leq\alpha.
\label{eq:ville}
\end{equation}
\label{thm:eprocess}
\end{theorem}

\begin{proof}
Under $H_0$, Proposition~\ref{prop:audit-probability} gives $\E[\eta_t\mid\text{pre-audit history}]\leq1$. Multiplication by the nonnegative preceding product establishes the supermartingale property, and Ville's inequality gives equation~\eqref{eq:ville} \citep{ville1939}.
\end{proof}

Rejecting $H_0$ supplies evidence against the assertion that every candidate is bad, although that evidence may come from a candidate earlier in the sequence. For example, consider the conservative factors $\eta_t=A_t(1-\rho)^{-n_t}$ with $\rho=0.05$ and threshold $800$. Earlier correct candidates can supply $130$ failure-free draws, giving $0.95^{-130}\approx786.92$. If the next candidate fails $26$ of $500$ cases and is audited once, a passing draw occurs with probability $474/500=0.948$ and raises the product to $0.95^{-131}\approx828.34$. The threshold is crossed while the current candidate is bad, because the earlier candidates already falsified the sequence hypothesis. Similarly, evidence that each family was satisfactory at some time need not identify one implementation satisfactory in every family. Theorem~\ref{thm:allocation} assigns evidence to the candidate being accepted and avoids this transfer across implementations.

\paragraph{Comparison of acceptance conditions.}
\label{sec:regimes}
Table~\ref{tab:regimes} lists sufficient conditions for an error target $\alpha$ over $T$ rounds and one family. Write $Q(n)=q_{N,\rho}(n)$ and $H_T=\sum_{t=1}^{T}\prod_{s=1}^{t-1}L_s$. The known-test condition is sharp for the unrestricted candidate class, and Theorem~\ref{thm:probe} supplies a constructive lower bound for repeated feedback from a hidden test set.

\begin{table}[t]
\centering
\small
\begin{tabular}{lll}
\toprule
Test selection & Sufficient size condition & New cases per round \\
\midrule
Known fixed tests & $n\geq N-\lfloor\rho N\rfloor$ & $0$ after setup \\
Hidden fixed tests & $H_T\,Q(n)\leq\alpha$ & $0$ after setup \\
Uniform generation & $T\,Q(m)\leq\alpha$ & $m$ \\
Fresh audit & $T\,Q(n)\leq\alpha$ & $n$ \\
\bottomrule
\end{tabular}
\caption{Sufficient acceptance conditions over a common horizon. The uniform-generation row assumes the sampling rule in Proposition~\ref{prop:dynamic}. Case counts refer to new samples and exclude repeated execution of regression tests.}
\label{tab:regimes}
\end{table}

Uniform generation and fresh auditing have the same sampling guarantee because both draw cases after the candidate is fixed. For targeted generation, the selection distribution determines the relation to population risk. For repeated testing on a hidden set, the feedback history determines the sufficient bound. These are properties of the sampling and information rules, which GTDD specifies separately from the model that generates development tests.

\section{Exact Numerical Analysis}
\label{sec:numerics}

We evaluate the finite-population formulas to compare the sample sizes needed for acceptance and the probabilities produced by repeated feedback. Exact ratios of binomial coefficients give both the false-acceptance bounds and the acceptance probabilities of candidates within tolerance. Values are calculated with integer arithmetic and rounded for display.

\subsection{Setup}

The calculations use populations of size $N=500$ with tolerance $\rho=0.05$ and error target $\alpha=0.05$. The smallest failure count above tolerance is $K_\rho=26$. For repeated decisions, we consider $M=8$ trajectories with $T=4$ rounds and $J=5$ families, giving $32$ opportunities to accept a repository or $160$ individual family conclusions.

The probing construction uses $r=3$ disjoint regions of size $m=15$ and at most $k=6$ reported cases per round. The sequence runs through all four prescribed rounds, including the terminal candidate, as assumed in Theorem~\ref{thm:probe}. When each region contains at most $6$ sampled cases, that candidate fails at least $3(15-6)=27$ population cases and passes the hidden test set. For acceptance within tolerance, we use the same family population and vary the candidate's total number of failures.

\subsection{Probability Calculation}

The audit sample size is the smallest integer $n$ for which $q_{500,0.05}(n)$ meets the applicable error allocation. The three allocations are $0.05$ for one candidate and family, $0.05/32$ for joint repository acceptance, and $0.05/160$ for simultaneous family conclusions. We compare each exact count with the sufficient count obtained from the exponential bound.

For the probing construction, let $c_h$ be the coefficient of $x^{h}$ in $\bigl(\sum_{i=0}^{k}\binom{m}{i}x^{i}\bigr)^{r}$. Counting how the sampled cases are distributed across the regions gives
\begin{equation}
p_{\mathrm{success}}=\frac{1}{\binom{N}{n}}\sum_{h=0}^{rk}c_h\binom{N-rm}{n-h}\,\1\bigl[rm-h\geq K_\rho\bigr],
\label{eq:probe-exact}
\end{equation}
where impossible binomial terms are zero. The indicator requires the terminal candidate's risk to exceed the tolerance. Appendix~\ref{app:proofs} gives the counting argument. For a candidate with a specified failure count, its acceptance probability follows directly from equation~\eqref{eq:hypergeometric}.

\subsection{Results}

The exact calculations quantify the effect of the acceptance target and the feedback available before the candidate is fixed. They also describe the probability that a zero-failure audit accepts a candidate whose population risk is already within tolerance.

\paragraph{Required sample size.}
Table~\ref{tab:sizes} gives the smallest exact audit size for each target. Simultaneous family conclusions require $131$ cases because $q_{500,0.05}(131)\approx0.000291773\leq0.05/160$, and the next smaller size gives approximately $0.000313825$. Joint repository acceptance reduces the multiplicity factor from $160$ to $32$ and requires $108$ cases per family. For one candidate and one family, $54$ cases suffice. The exact calculation improves on the exponential sufficient size in all three comparisons.

\begin{table}[t]
\centering
\small
\begin{tabular}{lrrr}
\toprule
Guarantee target & Exact $n$ & Exponential $n$ & $q_{500,0.05}(n)$ \\
\midrule
One candidate & $54$ & $59$ & $0.0472088$ \\
Repository over $32$ rounds & $108$ & $126$ & $0.00148353$ \\
$160$ family conclusions & $131$ & $158$ & $0.000291773$ \\
\bottomrule
\end{tabular}
\caption{Smallest exact audit sizes and exponential sufficient sizes. The repository row covers $8$ trajectories of $4$ rounds; the family row additionally covers $5$ separate family conclusions per round.}
\label{tab:sizes}
\end{table}

\paragraph{Adaptation through repeated feedback.}
Table~\ref{tab:probe} reports the lower bound and exact success probability of the three-region construction. The construction remains effective for sample sizes above the value $n^{*}$ obtained from the particular probe-size choice in equation~\eqref{eq:nstar}. At $n=131$, its success probability is approximately $0.8082$, while the same number of fresh audit cases gives an acceptance probability at most approximately $0.000292$ for a bad family. The first probability includes adaptation through three rounds of feedback, whereas the second concerns a candidate fixed before sampling. At $n=200$, the exact success probability remains $0.2195$ even though the union lower bound is zero.

\begin{table}[t]
\centering
\small
\begin{tabular}{rrr}
\toprule
Hidden test-set size $n$ & Union lower bound & Exact success \\
\midrule
$54$ & $0.9989$ & $0.9989$ \\
$100$ & $0.9504$ & $0.9511$ \\
$104$ & $0.9385$ & $0.9395$ \\
$108$ & $0.9245$ & $0.9260$ \\
$131$ & $0.7963$ & $0.8082$ \\
$158$ & $0.5199$ & $0.5876$ \\
$200$ & $0.0000$ & $0.2195$ \\
\bottomrule
\end{tabular}
\caption{Success of the three-region probe construction with $N=500$, $m=15$, and $k=6$. The probabilities describe the explicit casewise adversary after all four prescribed rounds.}
\label{tab:probe}
\end{table}

\paragraph{Acceptance within tolerance.}
Table~\ref{tab:operating} gives the probability that a candidate passes a zero-failure audit as its population failure count changes. A candidate with $5$ failures among $500$ cases has risk $0.01$ and passes a $131$-case audit with probability approximately $0.2174$, and at the tolerance boundary of $25$ failures that probability is approximately $0.000403$. The rule therefore combines strong false-acceptance control with a low acceptance probability for candidates near the tolerance boundary. An audit that permits failures moves this tradeoff at the cost of a larger sample: for joint repository acceptance at $0.05/32$, the zero-failure rule needs $108$ cases per family and accepts a candidate with $5$ failures with probability approximately $0.295$, whereas the rule that permits $4$ failures needs $221$ cases by Proposition~\ref{prop:audit-failures} and accepts the same candidate with probability approximately $0.984$.

\begin{table}[t]
\centering
\small
\begin{tabular}{rrrrr}
\toprule
Failures & Risk & Pass at $n=54$ & Pass at $n=108$ & Pass at $n=131$ \\
\midrule
$0$ & $0.000$ & $1.000000$ & $1.000000$ & $1.000000$ \\
$1$ & $0.002$ & $0.892000$ & $0.784000$ & $0.738000$ \\
$5$ & $0.010$ & $0.563334$ & $0.294558$ & $0.217358$ \\
$10$ & $0.020$ & $0.315393$ & $0.085553$ & $0.046396$ \\
$25$ & $0.050$ & $0.053264$ & $0.001920$ & $0.000403$ \\
$26$ & $0.052$ & $0.047209$ & $0.001484$ & $0.000292$ \\
\bottomrule
\end{tabular}
\caption{Exact probabilities that one family's zero-failure audit passes. The count $26$ is the first above tolerance.}
\label{tab:operating}
\end{table}

\section{Empirical Case Study}
\label{sec:case-study}

We evaluate how development tests change successive implementations of a stateful API. A paired design compares five policies from the same initial implementation, four added conditions refine that comparison, and every candidate is measured on a common private evaluation set. The experiment tests the development role of GTDD by comparing fixed tests with tests regenerated during repair, and it evaluates the acceptance rule afterwards on the recorded candidates. Four choices, each stated where it applies, distinguish this implementation from the protocol of Section~\ref{sec:gtdd}: the coding model answers single requests without tool access and so cannot edit protected files, candidates run inside the evaluator's process under a wall-clock limit, reduction preserves any disagreement with the reference rather than the designated one, and no acceptance rule acts during the four rounds.

\subsection{Experimental Setup}

The task is a single-client, in-memory key-value store modeled on a subset of Redis. Its stateful behavior allows counterexamples to expose interactions among commands and makes repeated implementation changes measurable on the same collection of traces.

\paragraph{Task and evaluation inputs.}
The interface has $61$ commands in eight families: keys, strings, lists, hashes, sets, sorted sets, expiry under a logical clock, and transactions with optimistic watches. A $2{,}072$-word specification fixes every reply and error code, including the order in which checks apply, and a $620$-line reference implementation computes expected replies. We selected the task because initial implementations fail on a substantial fraction of traces while counterexample-based repair remains possible. Each evaluation trace applies $10$ to $40$ operations to a fresh service, each generated trace contains at most $40$ operations, and a reduced trace is no longer than its original. The evaluation set consists of $2{,}000$ traces drawn once from a fixed random generator and used for every condition and round without being reported to either model. Each candidate is imported into the evaluator's process, which applies a limit of $0.5$ seconds of wall-clock time per trace. A reply is correct when it equals the reference reply as a JSON value, so a boolean never matches an integer. The harness that ran the trajectories, including the selection of counterexamples, compared replies with Python equality, under which a boolean matches an integer. Scoring every saved candidate again, each in a separate process under the same limit, reproduced every recorded count under that comparison exactly, and the JSON comparison changes the count of one candidate, the final module of one Dynamic blind trajectory, by $7$ of the $2{,}000$ traces. Every count reported below uses the JSON comparison. The changed module is the last of its trajectory, so every candidate that received feedback has the same count under both comparisons, and the counterexamples it received were selected under Python equality. The records of the five conditions consist of every candidate module, every generator the testing model wrote, the counterexamples reported in each round, and each candidate's failure count on the evaluation set after each round, together with the outcome of every evaluation trace from this second scoring.

\paragraph{Model and paired design.}
Both coding and testing use gpt-4.1-mini, snapshot 2025-04-14, through single requests without tool access. Temperature, nucleus sampling, and the seed are left at the API defaults, and a request that fails is retried up to three times. A coding request receives the specification, the current module, and the feedback, and returns a complete module that is committed before execution. Each of $30$ blocks begins with a module generated from the specification alone. All five conditions share that block's initial module and then perform three repair rounds, giving four coding rounds in total. Every condition receives the public example's outcome, and all trajectories run for the prescribed four rounds without an acceptance gate.

\paragraph{Statistical analysis.}
The outcome is the evaluation failure rate after each round. The primary comparison is the within-block difference after round four between Dynamic informed and Static LLM, and the remaining contrasts are secondary and exploratory, without adjustment for multiplicity. Percentile bootstrap intervals resample blocks $10{,}000$ times. The $p$-values come from a seeded Monte Carlo sign-flip test over blocks with $100{,}000$ sign assignments, under the null that within-block differences are symmetric about zero. Appendix~\ref{app:sensitivity} adds the median, the $10\%$ trimmed mean, and the exact sign test of every contrast, whose null concerns only the direction of the differences.

\subsection{Test Generation}

Each policy specifies when tests are generated, what information the testing model receives, and how many traces are sampled. The coding schedule and maximum number of counterexamples returned in a round are the same across conditions.

\paragraph{Baselines.}
Public only supplies the outcome of the public example and no additional feedback. Static property uses $400$ traces drawn before development from the evaluation generator with a different seed, so its tests share the evaluation distribution. These inputs remain fixed throughout the four coding rounds.

\paragraph{LLM-generated inputs.}
The testing model writes a trace generator so that sampled inputs can combine interacting rules. Every LLM tester is instructed to write a generator likely to expose violations, so every LLM condition, including the added ones below, is adversarial in the sense of Section~\ref{sec:gtdd}, and Dynamic informed is the only one that can aim at the committed candidate itself; Static property is not adversarial. Static LLM requests one generator from the specification before round one, samples $400$ traces, and fixes the valid inputs for subsequent rounds. Their mean number is $312.9$, and in $4$ of the $30$ blocks the generator produced no valid trace, so that condition received no development tests beyond the public example. Dynamic blind requests a new generator in every repair round from the specification and all previously reported counterexamples, then samples $200$ fresh traces, while Dynamic informed also supplies the current module to the testing model. Each dynamic policy therefore requests $600$ traces over the three repair rounds, compared with $400$ for Static LLM. The valid traces actually supplied are fewer, $13{,}816$ of $18{,}000$ for Dynamic blind and $14{,}760$ of $18{,}000$ for Dynamic informed, against $9{,}386$ of $12{,}000$ for Static LLM, because a generator can fail to run or produce invalid traces. A generator request is limited to $8{,}000$ output tokens, and a response cut off at that limit yields no runnable generator; a generator that runs can also raise an error or return only invalid traces. A round whose generator supplies no valid trace adds no new tests, which happened in $17$ and $13$ of the $90$ generator requests of the two dynamic policies. The comparison evaluates these policies as specified, combining per-round regeneration and access to earlier feedback with their respective sampling budgets.

\paragraph{Counterexample feedback.}
The four conditions with development tests return at most $6$ reduced counterexamples per round. They select by family coverage and then by length from a pool of $24$ reductions, and each coding request includes every counterexample reported so far. In this experiment, reduction preserves any disagreement with the reference, and the family label comes from the command at which the original trace first diverged. This criterion permits a reduction to expose a different disagreement from the original one. The resulting feedback is evaluated through the next candidate's failure rate on the common evaluation set.

\paragraph{Added conditions.}
After the five conditions had been analyzed, we ran four further conditions on the same blocks, first-round modules, and evaluation set, all on one day and in an order randomized within each block, to separate parts of the difference between Static LLM and the dynamic policies. Static LLM 600 writes one generator from the specification before round one, samples $600$ traces from it, and uses the valid ones from the start. Staged 600 uses the same generator and the same sampled traces, releasing $200$ of them in each repair round and keeping the valid ones. Regenerated without history writes a new generator from the specification alone in every repair round and samples $200$ traces, and Blind rerun repeats Dynamic blind on that day. Adjacent conditions are designed to differ mainly in one respect: the number of requested traces, whether the traces arrive at once or across rounds, whether one generator or three produce them, and whether the tester sees the reported counterexamples. These comparisons were chosen after the first results and are exploratory. For these conditions the records also hold every request and response, the original trace of every reported counterexample, the outcome of the public example, and the outcome of every evaluation trace.

\subsection{Results}

Both dynamic policies produce lower final failure rates than Static LLM. Table~\ref{tab:case-rates} shows the changes across coding rounds, and Table~\ref{tab:case-contrasts} reports the paired differences after round four.

\begin{table}[t]
\centering
\small
\setlength{\tabcolsep}{4pt}
\begin{tabular}{lrrrrr}
\toprule
Condition & Round 1 & Round 2 & Round 3 & Round 4 & New counterexamples in round 3 \\
\midrule
Public only & $0.548$ & $0.346$ & $0.342$ & $0.342$ & --- \\
Static property & $0.548$ & $0.396$ & $0.309$ & $0.298$ & $4.93$ \\
Static LLM & $0.548$ & $0.382$ & $0.363$ & $0.405$ & $3.40$ \\
Dynamic blind & $0.548$ & $0.360$ & $0.357$ & $0.279$ & $5.77$ \\
Dynamic informed & $0.548$ & $0.441$ & $0.316$ & $0.274$ & $5.37$ \\
\bottomrule
\end{tabular}
\caption{Mean evaluation failure rate after each coding round over $30$ blocks, and the mean number of counterexamples newly disclosed after round three, which the last repair round received. Round one is the shared first-round module.}
\label{tab:case-rates}
\end{table}

\begin{table}[t]
\centering
\small
\setlength{\tabcolsep}{4pt}
\begin{tabular}{lrrrr}
\toprule
Contrast (first minus second) & Difference & 95\% interval & $p$ & Lower and higher \\
\midrule
Dynamic informed $-$ static LLM & $-0.131$ & $[-0.228,-0.050]$ & $0.0015$ & $19,7$ \\
Dynamic informed $-$ dynamic blind & $-0.005$ & $[-0.044,+0.034]$ & $0.810$ & $12,15$ \\
Dynamic blind $-$ static LLM & $-0.126$ & $[-0.223,-0.044]$ & $0.0012$ & $20,8$ \\
Dynamic informed $-$ static property & $-0.024$ & $[-0.097,+0.043]$ & $0.534$ & $13,13$ \\
Static property $-$ static LLM & $-0.107$ & $[-0.205,-0.018]$ & $0.029$ & $20,7$ \\
Dynamic informed $-$ public only & $-0.067$ & $[-0.103,-0.032]$ & $0.0007$ & $20,8$ \\
Static LLM $-$ public only & $+0.064$ & $[-0.017,+0.159]$ & $0.195$ & $15,12$ \\
Static property $-$ public only & $-0.044$ & $[-0.102,+0.019]$ & $0.175$ & $21,5$ \\
Dynamic blind $-$ static property & $-0.019$ & $[-0.090,+0.046]$ & $0.611$ & $14,14$ \\
\bottomrule
\end{tabular}
\caption{Within-block differences of the failure rate after round four. The last column counts the blocks in which the first condition was lower and higher; the remaining blocks are ties. The first row is the primary comparison.}
\label{tab:case-contrasts}
\end{table}

\paragraph{Primary comparison.}
All conditions start from a mean failure rate of $0.548$. After round four, Dynamic informed has a mean rate of $0.274$, compared with $0.405$ for Static LLM. The within-block difference is $-0.131$, with a $95\%$ interval of $[-0.228,-0.050]$ and $p=0.0015$. Dynamic informed is lower in $19$ of the $30$ blocks. Dynamic blind also finishes below Static LLM, with a mean rate of $0.279$ and a paired difference of $-0.126$. The size of the mean difference depends heavily on three blocks in which Static LLM, but not the dynamic policies, ended with a module that failed almost every evaluation trace. In two of them the final repair broke a module that had passed most traces, once so that it could no longer be imported, and in the third the shared first-round module failed every trace and the Static LLM repairs never recovered it. The median within-block difference is $-0.044$ and the $10\%$ trimmed mean is $-0.076$, and the exact sign test over the $26$ blocks without a tie gives $p=0.029$ (Appendix~\ref{app:sensitivity}).

\paragraph{Source access and the property-based baseline.}
The estimated difference between Dynamic informed and Dynamic blind is $-0.005$, with interval $[-0.044,+0.034]$ and $p=0.810$. Source access produces no detectable additional improvement under this comparison. Neither dynamic condition has a statistically detectable advantage over Static property, which draws its inputs from the generator that produced the evaluation set and is therefore a strong baseline for this task.

\paragraph{Failures and successive repairs.}
Among the valid traces generated for the current module, $37.2\%$ fail under Dynamic informed and $35.7\%$ under Dynamic blind. In round three, these policies return means of $5.37$ and $5.77$ new counterexamples, respectively, compared with $3.40$ for Static LLM. The fixed LLM inputs continue to expose failures, but they supply fewer previously unreported counterexamples as development proceeds. Static LLM's mean failure rate rises between rounds three and four because of the two final repairs that broke their modules, and its final rate is not statistically distinguishable from Public only. Across all $150$ trajectories, one reaches a failure rate of at most $0.01$, so substantial errors remain after four rounds even in the conditions with the strongest average improvement.

\paragraph{Added conditions.}
After round four, the mean failure rates are $0.355$ for Static LLM 600, $0.311$ for Staged 600, $0.297$ for Regenerated without history, and $0.308$ for Blind rerun, against $0.405$ for Static LLM with $400$ requested traces (Appendix~\ref{app:sensitivity}). No adjacent difference is detectable, and neither is the difference between the two runs of Dynamic blind, $+0.029$. With equal numbers of requested traces on the same day, Blind rerun ends $0.047$ below Static LLM 600, with interval $[-0.121,+0.037]$. The rates fall over the first three steps and rise slightly at the last, without any single step being detectable with $30$ blocks, so this study does not attribute the advantage of the dynamic policies to regeneration, to the number of traces, or to the reported history alone.

\paragraph{Audits of the recorded candidates.}
This analysis uses the $600$ candidates of the five original conditions. The trajectories ran without an acceptance gate, but each recorded candidate's failure count on the evaluation set determines exactly the probability that a fresh audit drawn from that set would accept it (Proposition~\ref{prop:audit-failures}), provided that each round draws an independent sample and the audit affects only whether the trajectory stops, not the candidates that development would otherwise produce. We treat the $2{,}000$ evaluation traces as one population with tolerance $\rho=0.05$. Of the $600$ recorded candidates, $13$ are within this tolerance, including the final candidates of $6$ trajectories. Each block's first-round module is counted once per condition, since each condition audits it independently. A zero-failure audit that bounds false acceptance by $\alpha=0.05$ over the four rounds of one trajectory needs $83$ cases. Applied to all $150$ trajectories with stopping at the first acceptance, it yields $0.074$ expected false acceptances in total and $1.01$ expected acceptances of candidates within tolerance, and a trajectory ends without acceptance with mean probability $0.993$. A bound of $0.05$ over all $600$ audit opportunities requires $174$ cases, which lowers the expected false acceptances to $0.00027$ and the expected acceptances within tolerance to $0.251$. These values describe the audit rule applied to candidates that development actually produced; they concern the aggregate failure rate rather than individual families, and the public example is not part of the rule because the records do not include its outcome.

\section{Discussion}

The results assign complementary roles to generated tests, existing regression tests, and independent audits. Development uses concrete failures to improve the implementation, while acceptance requires a probability statement about the candidate currently under consideration.

\paragraph{Development feedback.}
The paired experiment shows that Dynamic blind and Dynamic informed end with lower mean failure rates than Static LLM, whose tests were generated once by the same model. Their generators use previous counterexamples and request more traces over the repair rounds, and the added conditions, which change mainly one of these features at each step, produce intermediate rates with no detectable step, so the result compares complete development policies. Static property draws its inputs from the generator that produced the evaluation set; its strength shows the value of inputs whose distribution matches the target, and the dynamic policies, which never see that generator, end without a detectable difference from it. GTDD accommodates both sources and saves their verified cases for subsequent regression checks.

\paragraph{Evidence for acceptance.}
Test visibility and feedback determine what an adaptive candidate can learn before acceptance. The exact known-test result describes unrestricted casewise adaptation, and the history bound accounts for the information returned by repeated testing on hidden inputs. Fresh conditional sampling after commitment gives a direct bound for the current candidate, with separate allocations for repository acceptance and simultaneous family conclusions. Over an unbounded sequence, summable allocations preserve that interpretation as implementations change. Counterexample reduction makes failures easier to repair, while its mapping to original cases determines how much information the feedback reveals.

\paragraph{The evaluation target.}
A fixed population gives an explicit target on which complete evaluation has no sampling error. The audit uses fewer cases to decide whether that target is met, and its sample size controls both false acceptance and the chance of accepting a candidate within tolerance. Table~\ref{tab:operating} makes this tradeoff visible for the zero-failure rule. For structured programs, the reference implementation defines the expected responses. For language and speech applications, the declared evaluation procedure defines the loss, so the choice of scenarios and judgments is part of the target itself. This formulation connects development tests to a precise acceptance claim.

\section{Conclusion}
We formulated GTDD as test-driven development in which a separate testing agent generates new inputs after each candidate is fixed and verified failures become feedback and regression tests. The finite-population analysis characterizes how known tests and repeated feedback affect false acceptance, and establishes fresh-audit guarantees for adaptive development. In the paired experiment, Dynamic blind and Dynamic informed ended with lower mean failure rates than Static LLM, further conditions did not attribute this difference to a single feature, source access produced no detectable additional gain, and property-based tests drawn from the evaluation distribution provided a strong baseline.

\bibliography{arXiv2.bbl}

\bibliographystyle{tmlr}

\clearpage
\appendix

\section{Event Order and Information Access}
\label{app:event-order}

The event order ensures that testing evaluates the committed implementation and that audit information cannot influence its selection. A source-informed round satisfies
\begin{equation}
\begin{split}
&\text{coding end}<\text{candidate commitment}<\text{testing start}<\text{testing end}\\
&\quad<\text{audit sampling}<\text{audit end}<\text{acceptance record}<\text{population evaluation end}.
\end{split}
\label{eq:event-order}
\end{equation}
Conditions without generated testing omit its events, and fixed tests are re-executed against the same committed snapshot. Audit randomness is independent of everything available before the draw, including the candidate and development-test outputs. The private audit record stores the allocation and its randomness for verification. When a cryptographic pseudorandom generator supplies the randomness, its security is an implementation assumption distinct from the exact uniform-sampling model.

Every candidate-readable file and dependency belongs to the committed manifest, while only the evaluator can access the reference implementation and audit records. Each case runs with a fresh writable layer that is discarded before the next case, and candidate output returns only through the permitted evaluation channel. Population evaluation leaves no persistent effect visible to a later candidate. A failure of these access restrictions or of the evaluator's own execution environment stops the affected run.

The aggregate acceptance bit is the only audit feedback to the coding session. An information bound for repeated testing on a hidden set counts this bit unless the other feedback determines it. Population statistics remain unavailable to the agents, including through scheduling, stopping decisions, or diagnostic messages.

\section{Counterexample Selection and Reduction}
\label{app:disclosure}

Counterexample selection determines which observed failures guide the next coding round. Both rules choose at most $k$ previously unreported development failures, using only the permitted information about those failures.

\paragraph{Selection rules.}
The standard rule first favors a semantic family not yet selected in the round, then a failure mechanism not yet selected, and finally shorter traces and canonical identifiers. The compact-failure rule first favors shorter traces, then fewer distinct operations, and then families already reported in an earlier round, with canonical identifiers breaking ties. Both rules choose deterministically among the same observed failures.

\paragraph{Reduction budget.}
A preliminary selection limits the reduction pool to $4k$ original traces, each containing at most $L$ operations, where $L$ is the schema's trace-length limit. Prefix search evaluates at most $L$ prefixes. Each pass of single-operation deletion evaluates at most $L$ variants and shortens the trace by one operation when successful, so at most $L$ passes are possible. The reduction cost is therefore at most $L+L^{2}$ candidate executions per original trace and $4k(L+L^{2})$ per round.

Every shorter trace is validated before execution and replaces the current trace only after it has been verified to exhibit the same designated failure. The selection rule is applied again to the reduced pool, and the chosen reductions are included in both the feedback and the regression tests. Private links to the original traces support analysis of which case identities the reductions reveal.

\section{Elementary Properties}
\label{app:elementary}

The following properties describe the information available to testing policies and the behavior measured by counterexample repair and regression testing. They separate the capabilities of a policy class from the performance of a particular testing model, and relate success on reported cases to the remaining population.

\paragraph{Information available to the tester.}
Let $\Pi^{B}_t$ contain testing policies that receive the contract and reported history without the current source, and let $\Pi^{I}_t$ contain policies that may also read the source. For a fixed candidate $f_t$, let $\Delta_t(\pi;f_t)$ be the expected number of valid nonduplicate failures that policy $\pi$ proposes under a common budget. If every source-blind policy can be implemented by ignoring the source, then $\Pi^{B}_t\subseteq\Pi^{I}_t$ holds, and we have $\sup_{\pi\in\Pi^{I}_t}\Delta_t(\pi;f_t)\geq\sup_{\pi\in\Pi^{B}_t}\Delta_t(\pi;f_t)$. This inclusion compares the best attainable performance within each policy class, and a fixed model and prompt need not attain either supremum.

\paragraph{Repair and regression.}
For nonempty $E_t$, the next-round repair rate is the fraction of reported cases that the next candidate passes, $1-|E_t|^{-1}\sum_{z\in E_t}\ell(f_{t+1};z)$, and it equals one when every reported case passes. If the candidate class can realize every binary loss pattern on $\calU\cup E_t$, passing all cases in $E_t$ imposes no restriction on the losses on $\calU\setminus E_t$, because the class contains a candidate with zero loss on $E_t$ and any prescribed losses on the complement. If an accepted candidate must pass every case in $G_{t-1}$, it has no regression on those cases. A rejected candidate may still regress, and the requirement imposes no constraint on cases outside $G_{t-1}$.

\paragraph{Agreement on reported cases.}
In the unrestricted casewise class, for any strict subset $D\subsetneq\calU$ and any prescribed binary losses on $D$, two candidates can incur those losses on $D$ and have different population risks: on the nonempty complement, one has zero loss on every case and the other has unit loss on every case. This is the construction in the proof of Proposition~\ref{prop:selection}, applied to the reported cases.

\section{Proofs for Repeated Testing}
\label{app:proofs}

The history bound and probing construction account for two different effects of feedback from a fixed hidden test set. The first counts possible candidate choices, while the second constructs choices that exploit the case identities revealed by feedback.

\begin{proof}[Proof of Theorem~\ref{thm:feedback-history}]
Condition on $\xi$, which leaves $S$ uniform over the size-$n$ subsets. For each history $\tau$ possible before round $t$, the candidate $f_\tau=\phi_t(\tau,\xi)$ is fixed independently of $S$. The round-$t$ false-acceptance event is contained in the union, over these histories, of the events that $f_\tau$ is bad and passes $S$. Proposition~\ref{prop:audit-probability} bounds each probability by $q_{N,\rho}(n)$. A union bound over histories and rounds, followed by averaging over $\xi$ and imposing the upper bound of one, gives the first inequality.

Each history has at most $L_s$ continuations at round $s$, giving $|\calT_{t-1}(\xi)|\leq\prod_{s=1}^{t-1}L_s$. The number of raw subsets of size at most $k$ is $\sum_{i=0}^{\min\{k,N\}}\binom{N}{i}\leq(N+1)^{k}$, and the acceptance bit contributes at most a factor of two. Bounding each history count by the largest and using $q_{N,\rho}(n)\leq(1-\rho)^{n}$ gives equation~\eqref{eq:history-count}. The sampling bound applies to each fixed candidate before conditioning on the realized feedback.
\end{proof}

\begin{proof}[Proof of Theorem~\ref{thm:probe}]
Fix disjoint regions $F_1,\ldots,F_r$ of size $m$, which exist because $rm\leq N$ holds. In round $s\leq r$, submit a candidate that fails exactly on $F_s$. All previously reported cases lie in earlier regions, so every failing sampled case in round $s$ is new to the feedback. On the event $G_s=\{|S\cap F_s|\leq k\}$, the rule returns every case in $S\cap F_s$.

In round $T$, fail on the union of the probed regions after removing the reported cases. On $G=\bigcap_{s=1}^{r}G_s$, this set is $\bigcup_{s=1}^{r}(F_s\setminus S)$ and is disjoint from the test set. The candidate passes with risk at least $r(m-k)/N>\rho$. Each region's occupancy has a hypergeometric distribution with parameters $(N,m,n)$, so a union bound gives $\Pr(G)\geq1-r\,p(m)$. Taking the maximum with zero gives the stated bound. The candidate choices depend only on the reported cases and the known population.
\end{proof}

\begin{proof}[Proof of Corollary~\ref{cor:nstar}]
The chosen $m$ gives mean $nm/N\leq k/2$, and the two integer conditions permit direct application of Theorem~\ref{thm:probe}. Without rounding, $r(kN/(2n)-k)>\rho N$ is equivalent to $1/n>2\rho/(rk)+2/N$, which gives equation~\eqref{eq:nstar}. If $n\geq rk/2$ holds, then we have $rm\leq rkN/(2n)\leq N$. The hypergeometric tail $p(m)$ depends on the specified values of $N$, $n$, $m$, and $k$.
\end{proof}

\begin{proof}[Derivation of equation~\eqref{eq:probe-exact}]
For region occupancies $h_1,\ldots,h_r$ satisfying $0\leq h_s\leq k$ for every $s$, the number of test sets realizing them is $\prod_{s=1}^{r}\binom{m}{h_s}\cdot\binom{N-rm}{n-\sum_{s=1}^{r}h_s}$. Summing the first product over occupancies with total $h$ gives $c_h$. The terminal candidate fails on $rm-h$ cases, so its risk exceeds tolerance exactly when $rm-h\geq K_\rho$ holds. Summing these counts and dividing by $\binom{N}{n}$ gives the formula.
\end{proof}

\section{Additional Case-Study Results}
\label{app:sensitivity}

This appendix reports the four added conditions of Section~\ref{sec:case-study} beside the two original conditions they refine, their within-block contrasts, and, for every contrast, the mean, median, trimmed mean, and exact sign test of the $30$ within-block differences.

\begin{table}[h]
\centering
\small
\begin{tabular}{lrrrrr}
\toprule
Condition & Round 1 & Round 2 & Round 3 & Round 4 & Valid tests \\
\midrule
Static LLM & $0.548$ & $0.382$ & $0.363$ & $0.405$ & $9{,}386$ \\
Static LLM 600 & $0.548$ & $0.377$ & $0.358$ & $0.355$ & $14{,}291$ \\
Staged 600 & $0.548$ & $0.352$ & $0.346$ & $0.311$ & $14{,}291$ \\
Regenerated without history & $0.548$ & $0.347$ & $0.304$ & $0.297$ & $15{,}381$ \\
Dynamic blind & $0.548$ & $0.360$ & $0.357$ & $0.279$ & $13{,}816$ \\
Blind rerun & $0.548$ & $0.396$ & $0.346$ & $0.308$ & $12{,}394$ \\
\bottomrule
\end{tabular}
\caption{Mean evaluation failure rate after each coding round over $30$ blocks for Static LLM, the four added conditions, and Dynamic blind, with the number of valid generated traces each condition received in total.}
\label{tab:extra-rates}
\end{table}

\begin{table}[h]
\centering
\small
\setlength{\tabcolsep}{3pt}
\begin{tabular}{lrrrr}
\toprule
Contrast (first minus second) & Difference & 95\% interval & $p$ & Lower and higher \\
\midrule
Staged 600 $-$ static LLM 600 & $-0.044$ & $[-0.132,+0.033]$ & $0.326$ & $15,9$ \\
Regenerated without history $-$ staged 600 & $-0.014$ & $[-0.083,+0.054]$ & $0.694$ & $17,12$ \\
Blind rerun $-$ regenerated without history & $+0.012$ & $[-0.040,+0.078]$ & $0.794$ & $15,13$ \\
Static LLM 600 $-$ static LLM & $-0.050$ & $[-0.163,+0.052]$ & $0.378$ & $16,11$ \\
Blind rerun $-$ dynamic blind & $+0.029$ & $[-0.027,+0.097]$ & $0.446$ & $11,17$ \\
Blind rerun $-$ static LLM 600 & $-0.047$ & $[-0.121,+0.037]$ & $0.282$ & $17,12$ \\
\bottomrule
\end{tabular}
\caption{Within-block differences of the failure rate after round four involving the added conditions. The last column counts the blocks in which the first condition was lower and higher.}
\label{tab:extra-contrasts}
\end{table}

\begin{table}[h]
\centering
\small
\setlength{\tabcolsep}{4pt}
\begin{tabular}{lrrrr}
\toprule
Contrast (first minus second) & Mean & Median & Trimmed mean & Sign test $p$ \\
\midrule
Dynamic informed $-$ static LLM & $-0.131$ & $-0.044$ & $-0.076$ & $0.029$ \\
Dynamic informed $-$ dynamic blind & $-0.005$ & $+0.001$ & $-0.006$ & $0.701$ \\
Dynamic blind $-$ static LLM & $-0.126$ & $-0.023$ & $-0.061$ & $0.036$ \\
Dynamic informed $-$ static property & $-0.024$ & $+0.000$ & $-0.009$ & $1.000$ \\
Static property $-$ static LLM & $-0.107$ & $-0.065$ & $-0.086$ & $0.019$ \\
Dynamic informed $-$ public only & $-0.067$ & $-0.029$ & $-0.061$ & $0.036$ \\
Static LLM $-$ public only & $+0.064$ & $-0.001$ & $+0.019$ & $0.701$ \\
Static property $-$ public only & $-0.044$ & $-0.059$ & $-0.060$ & $0.0025$ \\
Dynamic blind $-$ static property & $-0.019$ & $+0.000$ & $-0.006$ & $1.000$ \\
Staged 600 $-$ static LLM 600 & $-0.044$ & $-0.000$ & $-0.021$ & $0.307$ \\
Regenerated without history $-$ staged 600 & $-0.014$ & $-0.004$ & $-0.014$ & $0.458$ \\
Blind rerun $-$ regenerated without history & $+0.012$ & $-0.001$ & $-0.010$ & $0.851$ \\
Static LLM 600 $-$ static LLM & $-0.050$ & $-0.014$ & $-0.018$ & $0.442$ \\
Blind rerun $-$ dynamic blind & $+0.029$ & $+0.004$ & $+0.008$ & $0.345$ \\
Blind rerun $-$ static LLM 600 & $-0.047$ & $-0.011$ & $-0.047$ & $0.458$ \\
\bottomrule
\end{tabular}
\caption{Summaries of the $30$ within-block differences after round four. The trimmed mean removes the $3$ largest and $3$ smallest differences, and the exact sign test drops ties.}
\label{tab:sensitivity}
\end{table}